\documentclass[11pt]{amsart}

\usepackage{graphicx}
\usepackage{amssymb}
\usepackage{epstopdf}
\newcommand{\C}{{\mathbb C}}
\newcommand{\Z}{{\mathbb Z}}
\newcommand{\R}{{\mathbb R}}
\newcommand{\s}{{\mathbb S}}
\newcommand{\subarc}[1]{\widehat{#1}}
\newcommand{\sep}{\text{sep}}

\newtheorem{theorem}{Theorem}
\newtheorem{lemma}[theorem]{Lemma}
\newtheorem{corollary}[theorem]{Corollary}

\title[Polarization optimality of equally spaced points]{Polarization optimality of equally spaced points on the circle for discrete potentials} 
\author[D.P. Hardin]{Douglas P. Hardin}
\author[A.P. Kendall]{Amos P. Kendall}
\author[E.B. Saff]{Edward B. Saff}

\date{August 25, 2012}
\thanks{This research was supported, in part,
by the U. S. National Science Foundation under grants DMS-0808093 and DMS-1109266.  
}

\address{Center for Constructive Approximation\\
         Department of Mathematics\\
         Vanderbilt University\\
         1326 Stevenson Center\\
         Nashville, TN, 37240\\
         USA}
\email{doug.hardin@vanderbilt.edu}
\email{amos.p.kendall@vanderbilt.edu}
\email{edward.b.saff@vanderbilt.edu}



\begin{document}

\begin{abstract}
We prove a conjecture of Ambrus, Ball and  Erd\'{e}lyi that equally spaced points
maximize the minimum of  discrete potentials on the unit circle whenever the potential is of the form 
\begin{equation*} 
\sum_{k=1}^n f(d(z,z_k)),
\end{equation*}
where $f:[0,\pi]\to [0,\infty]$ is non-increasing and convex and $d(z,w)$ denotes the geodesic distance  between $z$ and $w$ on the circle. 
\end{abstract}

\maketitle

\section{Introduction and Main Results}
  Let $\s^1:=\{z=x+iy \mid x,y \in \mathbb{R},\,\,x^2+y^2  =1\}$ denote the unit circle in the complex plane $\C$.  For $z,w\in\s^1$, we denote by $d(z,w)$ the geodesic (shortest arclength) distance between $z$ and $w$.   Let $f:[0,\pi]\to [0,\infty]$ be non-increasing and convex on $(0,\pi]$ with $f(0)=\lim_{\theta\to 0^+}f(\theta)$.   It then follows that $f$ is a continuous extended real-valued function on $[0,\pi]$.  
   For    a list of $n$ points (not necessarily distinct) $\omega_n=(z_1 ,\ldots, z_n)\in (\s^1)^n$, we consider the $f$-{\em potential} of $\omega_n$,
\begin{equation}\label{Udef}
U^f(\omega_n;z):=\sum_{k=1}^n f(d(z,z_k))  \qquad (z\in \s^1),
\end{equation}
 and the {\em $f$-polarization} of $\omega_n$,
\begin{equation}
M^f(\omega_n;\s^1):=\min_{z\in \s^1}U^f(\omega_n;z).
\end{equation}
In this note, we are chiefly concerned with the {\em $n$-point $f$-polarization} of $\s^1$ (also called the  $n$th {\em $f$-Chebyshev constant} of $\s^1$),
\begin{equation}
M^f_n(\s^1):=\sup_{\omega_n\in (\s^1)^n}M^f(\omega_n;\s^1),
\end{equation}
which has been the subject of several recent papers (e.g., \cite{A}, \cite{ABE}, \cite{ES}, \cite{NR}).

 In the case (relating to Euclidean distance) when
 \begin{equation}\label{fs}
 f(\theta)=f_s(\theta):=  |e^{i\theta}-1|^{-s}=(2 \sin |\theta/2|)^{-s}, \, s>0,
 \end{equation}
 we abbreviate the notation for the above quantities by writing
  \begin{equation}
  \begin{split}
 U^s(\omega_n;z)&:=\sum_{k=1}^n f_s(d(z,z_k))=\sum_{k=1}^n \frac{1}{|z-z_k|^s}, \\
 M^s(\omega_n;\s^1)&:=\min_{z\in \s^1} \sum_{k=1}^n \frac{1}{|z-z_k|^s},\\
 M^s_n(\s^1)&:=\sup_{\omega_n \in (\s^1)^n} M^s(\omega_n;\s^1).\\
 \end{split}
 \end{equation}

  The main result of this note is the following theorem  conjectured by G. Ambrus et al  \cite{ABE}.  Its proof is given in the next section. 

 \begin{theorem}\label{polarS1Thm} Let $f:[0,\pi]\to [0,\infty]$ be non-increasing and convex on $(0,\pi]$ with $f(0)=\lim_{\theta\to 0^+}f(\theta)$.  If   $\omega_n$ is any configuration of $n$ distinct equally spaced points on $\s^1$, then $M^f(\omega_n;\s^1)=M^f_n(\s^1)$.  Moreover,  if the convexity condition is replaced by strict convexity, then such configurations are the only ones that achieve this equality. 
  \end{theorem}

Applying this theorem to the case of $f_s$ given in \eqref{fs} we immediately obtain the following.
\begin{corollary}\label{cor1} Let $s>0$ and $\omega_n^*:=\{ e^{i2\pi k/n}\,: k=1,2,\ldots,n\} $.  
If $(z_1,\ldots, z_n)\in(\s^1)^n$, then
\begin{equation}
\min_{z\in \s^1}\sum_{k=1}^n \frac{1}{|z-z_k|^s}\le M^s(\omega_n^*;\s^1)= M^s_n(\s^1),
\end{equation}
with equality if and only if $(z_1,\ldots, z_n)$ consists of distinct equally spaced points.
\end{corollary}

The
following representation of $M^s(\omega_n^*;\s^1)$ in terms of {\em Riesz $s$-energy} was observed in \cite {ABE}:
$$M^s(\omega_n^*;\s^1) = \frac{\mathcal{E}_s(\s^1;2n)}{2n} -
\frac{\mathcal{E}_s(\s^1;n)}{n},$$
where
$$\mathcal{E}_s(\s^1;n):=\inf_{\omega_n\in (\s^1)^n} \sum_{j=1}^n\sum_{\substack{k=1\\ k\neq j}}^n\frac{1}{|z_j-z_k|^s}.
$$
Thus, applying the asymptotic formulas for $\mathcal{E}_s(\s^1;n)$ given in \cite {BHS},  we obtain the dominant term of $ M^s_n(\s^1)$ as $n\to \infty$:
$$
M^s_n(\s^1)   \sim
\begin{cases}
 \displaystyle{\frac{2 \zeta(s)}{(2\pi)^s} \, (2^s - 1)n^s} \,, \quad
\enskip s> 1\,, \\ \\
\displaystyle { ({1}/{\pi}) \, n\log n}\,, \quad   \enskip s = 1\,, \\ \\
\displaystyle{\frac{2^{-s}}
{\sqrt{\pi}} \, \frac{\Gamma \big( \frac{1-s}{2}\big)}{\Gamma \big(1-\frac{s}{2} \big)}}\, n\,,
\quad \enskip s \in [0,1),
\end{cases}
  $$
where $\zeta(s)$ denotes the classical Riemann zeta function and  $a_n \sim b_n$ means that $\lim_{n \rightarrow \infty}{a_n/b_n} = 1$.  These asymptotics, but for $M^s (\omega_n^*;\s^1)$, were stated in \cite {ABE}\footnote{We remark that there is a
  factor of $2/(2\pi)^p$ missing in the asymptotics given in \cite{ABE} for the case $p:=s>1$.}.

 For $s$ an even
integer, say $s=2m$, the precise value of $M^{2m}_n(\s^1)=M^{2m}(\omega_n^*;\s^1)$ can be expressed in finite
terms, as can be seen from formula (1.20) in \cite {BHS}.

\begin{corollary} We have
$$M_n^{2m}(\s^1)= \frac{2}{(2\pi)^{2m}}
\sum_{k=1}^m n^{2k}\zeta(2k)\alpha_{m-k}(2m)(2^{2k}-1),\quad m \in \mathbb{N}, $$
where $\alpha_j(s)$ is defined via the power series for
$\text{\rm sinc\,} z = (\sin \pi z)/(\pi z):$
$$(\text{\rm sinc\,} z)^{-s} = \sum_{j=0}^\infty \alpha_j(s) z^{2j}\,,
\quad  \alpha_0(s)=1\,.$$
In particular,
\begin{equation*}
\begin{split}
M^2_n(\s^1)&=\frac{2}{(2\pi)^2}n^2\zeta(2)=\frac{n^2}{4},\\
 M_n^4(\s^1)  &=
\frac{2}{(2\pi)^4}[n^2\zeta(2)\alpha_1(4)(2^2-1)+n^4\zeta(4)(2^4-1)]=  \frac{n^2}{24}+\frac{n^4}{48} ,\\
 M_n^6(\s^1)  &=\frac{2}{(2\pi)^6}[n^2\zeta(2)\alpha_2(6)(2^2-1)+n^4\zeta(4)\alpha_1(6)(2^4-1)+n^6\zeta(6)(2^6-1)] \\ 
 & =\frac {n^2} {120} + \frac {n^4} {192} + \frac {n^6} {480},
 \end{split}
\end{equation*}
\end{corollary}
 The case $s=2$ of the above corollary was first
proved in \cite A,\cite {ABE} and the case $s=4$ was  first proved in \cite {ES}.
We remark that an alternative formula for $\alpha_j(s)$ is
$$\alpha_j(s) = \frac{(-1)^j B_{2j}^{(s)}(s/2)}{(2j)!}(2\pi)^{2j}, \qquad j = 0,1,2,\dots\,,$$
where $B_j^{(\alpha)}(x)$ denotes the generalized Bernoulli polynomial.  Asymptotic formulas for $M^f_n(\s^1)$ for certain other functions $f$ can be obtained from the asymptotic formulas given in 
\cite{BHS2}.

As other consequences of Theorem~\ref{polarS1Thm}, we immediately deduce 
that
equally spaced points are optimal for the following problems
\begin{equation}
\min_{\omega_n\in (\s^1)^n} \max_{z\in \s^1} \sum_{k=1}^n {|z-z_k|^\alpha},\qquad (0<\alpha\le 1),\end{equation}
and
  \begin{equation}
\label{logproblem}\max_{\omega_n\in (\s^1)^n} \min_{z\in \s^1} \sum_{k=1}^n \log \frac{1}{|z-z_k|},\end{equation}   
 with the solution to \eqref{logproblem} being well-known.
Furthermore, various generalizations of the polarization problem for Riesz potentials for configurations on $\s^1$ are worthy of consideration, such as minimizing the potential on circles concentric with $\s^1$.  


\section{Proof of Theorem~\ref{polarS1Thm}}

For distinct points $z_1,z_2\in \s^1$, we let $\subarc{z_1z_2}$ denote the closed subarc of $\s^1$ from $z_1$ to $z_2$ traversed in the counterclockwise direction. We further let $\gamma(\subarc{z_1z_2})$ denote the length of $\subarc{z_1z_2}$ (thus, $\gamma(\subarc{z_1z_2})$ equals either $d(z_1,z_2)$ or $2\pi-d(z_1,z_2)$).   Observe that the points $z_1$ and $z_2$ partition $\s^1$ into two subarcs: $\subarc{z_1z_2}$ and $\subarc{z_2z_1}$.   The following lemma (see proof of Lemma 1 in \cite{ABE}) is a simple consequence of the convexity and monotonicity of the function  $f$
 and is used    to show that any $n$-point configuration $\omega_n\subset \s^1$ such that $M^f(\omega_n;\s^1)=M_n^f(\s^1)$ must have the property that any local minimum of $U^f(\omega_n; \cdot)$ is a global minimum of this function. 
 
   \begin{figure}[htbp]
\begin{center}
\vspace{.15in}

\includegraphics[scale=.7]{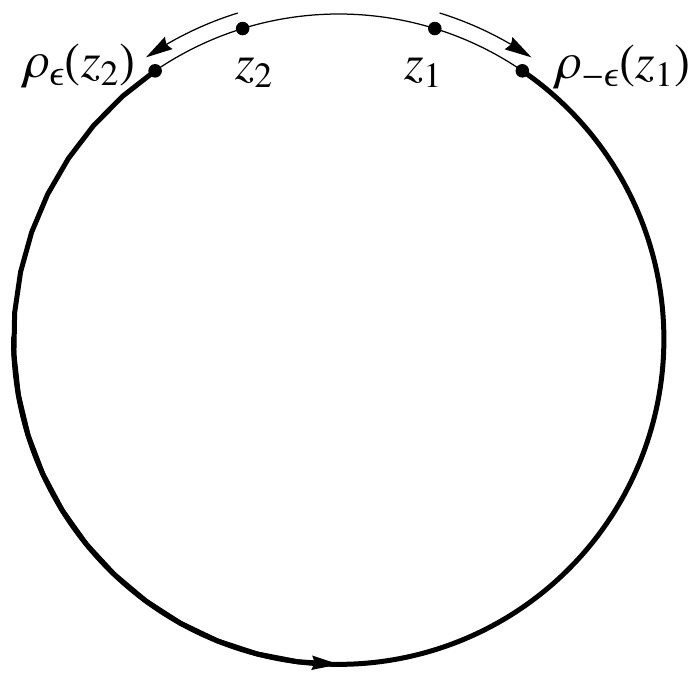}

\vspace{-2.1in}  
$\subarc{z_1 z_2}$

\vspace{1.5in}
$\subarc{\rho_\epsilon(z_2) \rho_{-\epsilon}(z_1)}$

\vspace{.2in}

\caption{The points $z_1, z_2, \rho_{-\epsilon}(z_1), \rho_{\epsilon}(z_2)$ in Lemma~\ref{DeltaLemma}.  The potential increases at every point in the subarc $\subarc{\rho_\epsilon(z_2) \rho_{-\epsilon}(z_1)}$ when $(z_1, z_2) \to (\rho_{-\epsilon}(z_1), \rho_{\epsilon}(z_2))$; see \eqref{DeltaIneq}. }
\label{Fig1}
\end{center}
\end{figure}

  For $\phi\in \R$ and $z\in \s^1$, we let $\rho_\phi(z):=e^{i\phi}z$ denote the counterclockwise rotation of $z$ by the angle $\phi$.

\begin{lemma}[\cite{ABE}] \label{DeltaLemma}    Let $z_1, z_2\in \s^1$ and $0<\epsilon<\gamma(\subarc{z_2z_1})/2$. Then
with $f$ as in Theorem~\ref{polarS1Thm},
\begin{equation} \label{DeltaIneq}
U^f((z_1,z_2);z)\le U^f((\rho_{-\epsilon}(z_1),\rho_\epsilon(z_2));z),
\end{equation}
for $z$ in the subarc  $\subarc{\rho_\epsilon(z_2) \rho_{-\epsilon}(z_1)},$
while the reverse inequality holds for $z$ in the subarc  $\subarc{z_1 z_2}$.    If $f$ is strictly convex on $(0,\pi]$, then these inequalities are strict.  If $z_1=z_2$, then we set $\subarc{z_1 z_2}=\{z_1\}$ and $\subarc{z_2 z_1}=\s^1$.
\end{lemma}

We now assume that  $\omega_n=(z_1 ,\ldots, z_n)$ is ordered in a counterclockwise manner and also that the indexing is extended periodically so that $z_{k+n}=z_k$ for $k\in\Z$.  For $1\le k\le n$ and $\Delta\in \R$, we define  $\tau_{k,\Delta}:(\s^1)^n\rightarrow(\s^1)^n$ by
$$\tau_{k,\Delta}(z_1,\ldots, z_k,z_{k+1},\ldots, z_n):=(z_1,\ldots, \rho_{-\Delta}(z_k),\rho_{\Delta}(z_{k+1}),\ldots, z_n).$$
If $z_{k-1}\neq z_k$  and $z_{k+1}\neq z_{k+2}$,  then $\tau_{k,\Delta}(\omega_n)$ retains the ordering of $\omega_n$   for $  \Delta $  positive and sufficiently small.
Given $\boldsymbol{\Delta}:=(\Delta_1,\ldots, \Delta_n)^T\in \R^n$, let $\tau_{\boldsymbol{\Delta}}:=\tau_{n,\Delta_n}\circ \cdots \circ \tau_{2,\Delta_2}\circ \tau_{1,\Delta_1}$ and
$\omega_n':=\tau_{\boldsymbol{\Delta}} (\omega_n).$
   Letting  $\alpha_k:=\gamma(\subarc{z_k z_{k+1}})$ and  $\alpha_k':=\gamma(\subarc{z_k' z_{k+1}'})$ for $k=1,\ldots, n$,
  we obtain the system of $n$ linear equations:
  \begin{equation}\label{system}
  \alpha'_k=\alpha_k-\Delta_{k-1}+2\Delta_k-\Delta_{k+1}, \qquad (1\le k\le n),
  \end{equation}
 which is satisfied as long as $\sum_{k=1}^n\alpha_k'=2\pi$    or, equivalently,  if $\omega_n'$ is ordered counterclockwise.  
Let  $$\sep(\omega_n):=\min_{1\le \ell \le n} \alpha_\ell.$$   Then   \eqref{system} holds if 
\begin{equation}\label{DeltaSuffCond}
\max_{1\le k\le n}|\Delta_k|\le (1/4)\sep(\omega_n),
\end{equation} 
in which case, 
  the configurations 
  \begin{equation}\label{omeganell}
  \omega_{n,\boldsymbol{\Delta}}^{(\ell)}:=\tau_{n,\Delta_\ell}\circ \cdots \circ \tau_{2,\Delta_2}\circ \tau_{1,\Delta_1}(\omega_n),\qquad(\ell=1,\ldots,n)
  \end{equation} are all ordered counterclockwise.   If the components of $\boldsymbol{\Delta}$ are nonnegative, then we may replace the `(1/4)' in \eqref{DeltaSuffCond} with `(1/2)'.
  
 \begin{lemma}\label{Delta*}
  Suppose $\omega_n=(z_1 ,\ldots, z_n)$ and $\omega_n'=(z'_1 ,\ldots, z'_n)$ are $n$-point configurations on $\s^1$ ordered in a counterclockwise manner. Then there is a unique $\boldsymbol{\Delta}^*=(\Delta_1^*,\ldots, \Delta_n^*)\in \R^n$ so that
  \begin{itemize}
  \item[(a)] $\Delta_k^*\ge 0$, $k=1,\ldots, n$,
  \item[(b)] $\Delta_{j }^*=0$ for some $j \in\{1,\ldots, n\}$, and
  \item[(c)] $\tau_{\boldsymbol{\Delta^*}}(\omega_n)$ is a rotation of $\omega_n'$.
  \end{itemize}
  \end{lemma}
 \begin{proof}

     The  system   \eqref{system} can be expressed in the form 
     \begin{equation}\label{ADB}A\boldsymbol{\Delta}=\boldsymbol{\beta},
     \end{equation}
 where
$$A:=\begin{pmatrix} 2 & -1& 0 & 0 & \cdots & -1\\-1 & 2 & -1&0& \cdots & 0 \\
\vdots &   &  & &  & \vdots\\
0 & 0 &   \cdots & -1 &2 &-1\\
-1 & 0 &   \cdots &0& -1 &2\end{pmatrix}, \quad \boldsymbol{\Delta}:=\begin{pmatrix} \Delta_1 \\ \Delta_2\\ \vdots \\   \\ \Delta_n \end{pmatrix},\text{ and } \boldsymbol{\beta}:=\begin{pmatrix} \alpha_1'-\alpha_1\\ \alpha_2'-\alpha_2\\ \vdots \\   \\ \alpha_n'-\alpha_n \end{pmatrix}.
$$
It is elementary to verify that   $\ker A=(\text{range }A)^\perp= \text{span } (\boldsymbol{1})$, where $\boldsymbol{1}=(1,1,\ldots,1)^T$.     Since $\boldsymbol{\beta}^T\boldsymbol{1}=\sum_{k=1}^n (\alpha_k'-\alpha_k)=0$, the linear system \eqref{ADB} always has a solution $\boldsymbol{\Delta}$.  Let $j\in\{1, \ldots, n\}$ satisfy $\Delta_{j}= \min_{1\le k\le n} \Delta_k  $.
Then subtracting $\Delta_{j}  \boldsymbol{1}$ from $\boldsymbol{\Delta}$, we obtain the desired $\boldsymbol{\Delta}^*$.  Since $\ker A=\text{span }\boldsymbol{1}$,  there is at most one solution of \eqref{ADB} satisfying properties (a) and (b), showing that $\boldsymbol{\Delta^*}$ is unique.

Part (c) holds as a direct result of the fact that both  $\omega_n$ and $\omega_n'$ are  ordered counterclockwise.

 \end{proof}

\begin{lemma} \label{lemmaell} Let $\Omega_n=(z_1,\ldots, z_n)$ be a configuration of $n$ distinct points on $\s^1$ ordered counterclockwise,  and with $f$ as in Theorem~\ref{polarS1Thm}, 
suppose $\boldsymbol{\Delta}=(\Delta_1,\ldots, \Delta_n)\in \R^n$ is such that   
\begin{itemize}
\item[(a)]  $0\le \Delta_k \le (1/2) {\rm sep}(\Omega_n)$ for $k=1,\ldots, n$, and 
\item[(b)]   there is some $j\in\{1,\ldots, n\}$ for which $\Delta_j=0$.
\end{itemize}
  Let $\Omega_n':=\tau_{\boldsymbol \Delta}(\Omega_n)=(z'_1,\ldots, z'_n)$.  Then 
$\subarc{z'_j z'_{j+1}}\subset \subarc{z_jz_{j+1}}$ and 
\begin{equation}\label{U1}
U^f(\Omega_n;z) \le U^f(\Omega_n';z)  \qquad (z\in \subarc{z_j'z_{j+1}'}).
\end{equation}
If $f$ is strictly convex on $(0,\pi]$ and $\Delta_k>0$ for at least one $k$, then the inequality \eqref{U1} is strict.
\end{lemma}

We remark that $\Delta_k=0$ for all $k=1,\ldots, n$ is equivalent to saying that the points are equally spaced. 

\begin{proof}
Recalling \eqref{omeganell}, it follows from condition (a) that  $(z^{(\ell)}_1,\ldots, z^{(\ell)}_n):= \omega_{n,\boldsymbol{\Delta}}^{(\ell)}$ are counterclockwise ordered.  Since  $\Delta_j=0$ and $\Delta_k\ge 0$ for $k=1,\ldots,n$,  the points $z^{(\ell)}_j$ and $z^{(\ell)}_{j+1}$ are moved at most once as $\ell$ varies from 1 to $n$ and move toward each other, while  remaining in the complement of all other subarcs  $\subarc{z^{(\ell)}_k z^{(\ell)}_{k+1}}$, i.e., 
 $$\subarc{z'_jz'_{j+1}}=\subarc{z^{(n)}_jz^{(n)}_{j+1}}\subseteq \subarc{z^{(\ell)}_jz^{(\ell)}_{j+1}}\subseteq \subarc{z^{(\ell)}_{k+1}z^{(\ell)}_k},$$ for $k\in\{1,\ldots,n\}\setminus\{j\}$ and $\ell\in\{1,\ldots, n\}$.  
  Lemma~\ref{DeltaLemma} implies that, for $\ell=1,\ldots, n$, we have $U^f(\omega^{(\ell-1)}_n;z) \le U^f(\omega^{(\ell)}_n;z)$ for $z\in \subarc{z^{(\ell)}_j z^{(\ell)}_{j+1}}$ (where  $\omega^{(0)}_n:=\omega_n$) and the inequality is strict if $\Delta_\ell>0$.  Hence,   \eqref{U1} holds and the inequality is strict if $f$ is strictly convex and $\Delta_k>0$ for some $k=1,\ldots, n$. 
 \end{proof}
 
We now proceed with the proof of Theorem~\ref{polarS1Thm}.  Let $\omega_n=(z_1,\ldots, z_n)$ be a non-equally spaced configuration of $n$  (not necessarily distinct) points on $\s^1$ ordered counterclockwise.  By Lemma~\ref{Delta*}, there is some equally spaced configuration $\omega_n'$ (i.e., $\alpha'_k=2\pi/n$ for $k=1,\ldots, n$) and some $\boldsymbol{\Delta^*}=(\Delta_1^*,\ldots, \Delta_n^*)$ such that (a) $\omega_n'=\tau_{\boldsymbol{\Delta^*}}(\omega_n)$, (b) $\Delta^*_k\ge 0$ for $k=1,\ldots,n$, and (c) $\Delta^*_j=0$ for some $j\in\{1,\ldots,n\}$.      Then \eqref{system} holds with $\alpha_k:=\gamma(\subarc{z_k,z_{k+1}})$ and $\alpha'_k:=2\pi/n$.    Since $\omega_n$ is not equally spaced, we have  $\Delta^*_k>0$ for at least one value of $k$. 

For $0\le t\le 1$, let $\omega_n^t:=\tau_{(t\boldsymbol{\Delta^*})}(\omega_n)=(z_1^t,\ldots, z_n^t)$ and, for $k=1,\ldots, n$, let  $\alpha_k^t:=\gamma(\subarc{z_k^t z_{k+1}^t})$. Recalling \eqref{system},   observe that
\begin{align*}\alpha_k^t&=\alpha_k-t(\Delta_{k-1}+2\Delta_k-\Delta_{k+1})\\ &=\alpha_k+t(2\pi/n-\alpha_k)\\ &=(1-t)\alpha_k+t(2\pi/n),
\end{align*} for $0\le t\le 1$ and $k=1,\ldots, n$, and so
$\sep(\omega_n^t)\ge t(2\pi/n)$.      Now let  $0<t<s<\min(1,t(1+\pi/(nD)))$, where $D:=\max \{\Delta_k : 1\le k\le n\}$.  Then Lemma~\ref{lemmaell} (with $\Omega_n=\omega_n^t$, $\boldsymbol{\Delta}=(s-t)\boldsymbol{\Delta^*}$, and $\Omega_n'=\tau_{\boldsymbol{\Delta}}(\Omega_n)=\omega_n^s$) implies that 
$\subarc{z_j^s z_{j+1}^s}\subseteq\subarc{z_j^t z_{j+1}^t}$ and that
\begin{equation}\label{U2}
U^f(\omega_n^t;z) \le U^f(\omega_n^s;z)  \qquad (z\in \subarc{z_j^sz_{j+1}^s}),
\end{equation}
where the inequality is sharp if $f$ is strictly convex.

Consider  the function  $$h(t):=\min\  \{U^f(\omega_n^t;z):  z\in\subarc{z_j^t z_{j+1}^t}\}, \qquad (0\le t\le 1).$$  Observe that $$h(t)\le \min\ \{U^f(\omega_n^t;z): {z\in\subarc{z_j^s z_{j+1}^s}}\}\le \min \{U^f(\omega_n^s;z): {z\in\subarc{z_j^s z_{j+1}^s}}\}= h(s),$$ for  $0<t<s<\min(1,t(1+\pi/(nD)))$.  It is then easy to verify that  $h$ is  non-decreasing  on $(0,1)$. Since $\omega_n^t$ depends continuously on $t$,  the function $h$ is continuous on $[0,1]$ and thus $h$ is   non-decreasing on $[0,1]$.  

We then obtain the desired inequality
 $$M^f(\omega_n;\s^1)\le h(0) \le h(1)=M^f(\omega_n';\s^1),$$  where the last equality is a consequence of the fact that  $\omega_n'$ is an equally spaced configuration and so the minimum of $U^f(\omega_n';z)$ over $\s^1$ is the same as the minimum over $\subarc{z_j'z_{j+1}'}$.   If $f$ is strictly convex, then $h(0)<h(1)$ showing that any optimal   $f$-polarization configuration must be equally spaced.  This completes the proof of Theorem~\ref{polarS1Thm}.  \hfill $\Box$
 
 \bigskip
 
\noindent{\bf Acknowledgements:} We thank the referees for their helpful suggestions  to improve the manuscript.


\begin{thebibliography}{99}
 \bibitem{A} G. Ambrus, Analytic and Probabilistic Problems in Discrete Geometry,
Ph.D. Thesis, University College London, 2009.

 \bibitem{ABE} G. Ambrus, K. Ball, and T. Erd\'{e}lyi, Chebyshev constants for the unit circle, {\em Bull. London Math. Soc.} {\bf 45}(2) (2013), 236--248.
 \bibitem{BHS} J.S. Brauchart, D.P. Hardin, and E.B. Saff, The {R}iesz energy of the {$N$}th roots of unity: an asymptotic expansion for large {$N$},
{\em Bull. London Math. Soc.}, \textbf{41} (4) (2009), 621--633.
 \bibitem{BHS2} J.S. Brauchart, D.P. Hardin and E.B. Saff, Discrete energy asymptotics on a Riemannian circle, {\em Uniform Distribution Theory}, \textbf{6} (2011), 77--108.

 \bibitem{ES} T. Erd\'{e}lyi and E.B. Saff, Riesz polarization inequalities in higher dimensions, (submitted), J. Approx. Theory {\bf 171} (2013), 128--147.
  

\bibitem{NR} N. Nikolov and R. Rafailov, On the sum of powered distances to certain sets of points on the circle, {\em Pacific J. Math.},  \textbf{253}(1), (2011), 157--168.

 \end{thebibliography}
\end{document}